\documentclass[22,reqno]{amsart}
\usepackage{hyperref}

\newtheorem{theorem}{Theorem}[section]
\newtheorem{proposition}[theorem]{Proposition}
\newtheorem{lemma}[theorem]{Lemma}
\newtheorem{remark}[theorem]{Remark}

\numberwithin{equation}{section}

\DeclareMathOperator{\tr}{tr}

\newcommand{\Es}{\mathbb{E}}
\newcommand{\eps}{\epsilon}
\newcommand{\dd}{\mathrm{d}}
\newcommand{\h}{\mathcal{H}}
\newcommand{\R}{\mathbb{R}}
\newcommand{\E}{\mathcal{E}}
\newcommand{\G}{\mathcal{G}}

\newcommand{\z}{\mathbb{Z}}
\newcommand{\N}{\mathbb{N}}
\newcommand{\M}{\mathcal{M}}
\newcommand{\p}{\mathcal{P}}
\newcommand{\lp}{\mathcal{L}}
\newcommand{\cc}{\mathcal{C}}

\newcommand{\la}{\langle}
\newcommand{\ra}{\rangle}

\newcommand{\e}{\mathrm{e}}

\newcommand{\abs}[1]{\left\lvert #1 \right\rvert}
\newcommand{\norm}[1]{\left\lVert #1 \right\rVert}
\newcommand{\normsch}[1]{{\left\lVert #1 \right\rVert}_2}
\newcommand{\scal}[1]{\la #1 \ra}
\newcommand{\f}{\mathcal{X}}
\newcommand{\dirac}{\delta_}
\newcommand{\style}{\displaystyle}

\title[Spectral properties of dynamical localization]{Spectral properties of dynamical localization for Schr\"odinger Operators}

\author{Fran\c cois Germinet}
\address{ Universit\'e de Cergy-Pontoise, CNRS UMR 8088, IUF,
D\'epartement de Math\'ematiques,
F-95000 Cergy-Pontoise, France}
\email{francois.germinet@u-cergy.fr}

\author{Amal Taarabt}
\address{ Universit\'e de Cergy-Pontoise,CNRS UMR 8088,
D\'epartement de Math\'ematiques,
F-95000 Cergy-Pontoise, France}
 \email{amal.taarabt@u-cergy.fr}

\thanks{The authors are supported by the grant
  ANR-08-BLAN-0261-01.  The authors would also like to thank the Centre
  Interfacultaire Bernoulli (EPFL, Lausanne) for their hospitality.} 
  
\begin{document}
\begin{abstract}
We investigate the equivalence between dynamical localization and localization properties of eigenfunctions of Schr\"odinger Hamiltonians.  
We introduce three classes of equivalent properties and study the relationships between them. These relationships are shown to be optimal 
thanks to counter examples.
\end{abstract}

\maketitle

\section{Introduction}

In this note we investigate the equivalence between dynamical localization and localization properties of eigenfunctions for 
self-adjoint operators on a Hilbert spaces $\h$. Although the analysis holds in greater generality we are mainly interested in the 
particular cases where $\h=\mathrm{L}^2(\R^d)$ or $\ell^2(\z^d)$. We shall nevertheless extend the discussion to operators on more general
graphs. Our motivation comes from ergodic Schr\"odinger operators and more precisely from
random and quasi-periodic Schr\"odinger operators, where dynamical localization has been proved, that is the non spreading of wave-packets 
under the time evolution coming from the Schr\"odinger equation \cite{A,GDB,JL,G,DS,GK2,GJ,BJ, GK3}.

Although in the context of Anderson models, localization has
been interpreted as pure point spectrum with exponentially localized eigenfunctions,
it is by now well established that the latter is not sufficient to ensure dynamical localization, even with a uniform finite 
localization length, so that Del Rio, Jitomirskaya, Last and Simon raised the following natural question: what is localization ? \cite{DRJLS1,DRJLS2, GKT}. Actually even a single energy can be responsible for a nontrivial transport \cite{JSS,DLS}.

To go beyond Anderson localization, stronger forms of localization properties have been introduced in order to derive dynamical 
localization \cite{DRJLS1,DRJLS2,G,GK2}. Note that if an eigenfunction $\phi$ decays as 
$|\phi(x)|\leq C_\phi \ \e^{-\sigma|x-x_\phi|}$, then we can only conclude that
$|\phi(x)|\le \frac12$  if $|x-x_\phi|\geq
\frac1\sigma \log(2C_\phi)$, suggesting that not only the localization
length  $\frac 1 \sigma$ is of importance, but that  the constant $C_\phi$ also matters. This trivial observation, combined to the 
fact that when  dense point spectrum is observed a given non trivial wave-packet contains an infinite number of eigenfunctions, 
suggests that a better control on the exponential decay of eigenfunctions is required if we want to go beyond pure spectral results. 
In particular such a better control should ensure summability of the contributions of a (possible) infinite number of eigenfunctions. 
Two families of localization properties have been introduced in the literature which both solve this problem: \emph{semi-uniformly 
localized eigenfunctions} (SULE) where the constant $C_\phi$ is explicit in the center $x_\phi$ \cite{DRJLS1,DRJLS2} and 
\emph{semi-uniform decay of eigenfunction correlations} (SUDEC) where the two-sites eigenfunction correlation function is controlled 
in a summable way in energy \cite{G,GK2}. (SULE) and (SUDEC) properties are shown to be equivalent in great generality and to imply 
dynamical localization.

Besides the physical issue of controlling the time evolution of wave-packets, dynamical localization (DL) and the properties (SULE) 
and (SUDEC) have been shown to play a crucial role in the mathematical proof of several phenomenon of physical interest, like Mott 
formula \cite{KLM}, the quantum Hall effect \cite{H,Be} (including the quantification of the  Hall conductance \cite{BES}, the 
existence of plateaux due to localized states \cite{AG,GKS1,GKS2}, the validity of the Kubo formula \cite{BGKS}, the equality of 
the bulk / edge conductances when the Fermi level lies within a region of localization \cite{EGS,Ta}, the regularization of the 
edge conductance \cite{CG,EGS}). Also, these properties turn out to be a key ingredient in order to get relevant informations about
 the statistics of the eigenvalues of the Anderson model: finite multiplicity \cite{GK2,GK3}, simplicity of the spectrum \cite{KM}, 
Poisson statistics \cite{Mi,GKl1,GKl2}, asymptotic eigenvalues ergodicity \cite{Klo}, level spacings statistics \cite{GKl1,GKl2}.

In this article, we come back to these properties that have been established and used in the mathematical physics literature over
the past 20 years or so. We show that they can be gathered in three classes of equivalent properties and we study the relationships
between them. The first class corresponds to dynamical localization, the second one to (SULE) and (SUDEC) for a basis of eigenvectors, 
and the third one to a stronger form of (SULE) and (SUDEC) where these properties hold for all vectors in the range of the eigenprojector
(SULE+) and (SUDEC+).  This last strong form of localization actually implies finite multiplicity of the eigenvalues.
It was commonly believed that localization properties of eigenfunctions like (SULE) or (SUDEC) were stronger notions of localization than the non spreading of wave-packets (DL). However, the results available in \cite{DRJLS2,T} indicate that this naive picture is 
not quite right  and that detailed informations on the decay of eigenfunctions can be derived from the boundeness of moments of wave-packets.
In this note, we extend the results of \cite{DRJLS2} and \cite{T} and present a clean picture of the situation. In particular, it solves 
an open question raised in \cite{DRJLS2} about the equivalence between (DL) and (SULE), and the role payed by the multiplicity of the eigenvalues. We further extend the analysis to graphs or trees with moderate growth.

The paper is organized as follows. In section~\ref{results} we introduce the (DL), (SULE) and (SUDEC) properties and state our main 
results. In Section~\ref{proofs} we present the details of the proofs. In Section~\ref{counterex} we provide counter examples, 
showing that our results are optimal. Appendix A contains  the proof of technical lemma, and in Appendix B we extend the first 
result of Section~\ref{results} to the random case. 

\section{Main results}\label{results}
We consider a self-adjoint operator $H$ on the Hilbert space $\h=\mathrm{L}^2 (\R^d)$. The case $\ell(\z^d)$ is slightly simpler. 
At the end of this section we extend the results to graphs.

Given $x\in\R^d$, we set $|x|:= \max\{|x_1|,|x_2|,\dots,|x_d|\}$. We use $|X_u|$ to denote the operator given by the multiplication 
by the function $|x-u|$. By $\Lambda_L(x)$ we denote the open box centered at $x\in\z^d$ with length side $L>0$
and we write $\chi_{x,L}$ for its the characteristic function and set $\chi_x :=\chi_{x,1}$. 
Given an open interval $I \subset\R$, we consider $\mathcal{C}^{\infty} _{c,+}(I)$ is the class of nonnegative 
real valued functions infinitly differentiable with compact support 
contained in $I$. The notation $\normsch{A}$ corresponds to the Hilbert-Schmidt norm of the operator $A$.
We set $P_{E}:=\chi_{\{E\}}(H)$ the spectral projection associated to $E\in\R$.

For a given $\sigma>0$ and $\zeta\in (0,1]$, we introduce
\begin{align}\label{moment} 
M_u(\sigma,\zeta,\f,t):= &\normsch{\e^{{\frac{\sigma}{2}}\abs{X_u}^{\zeta}} \e^{-itH} \f(H) \chi_u}^2 \\
=&  \tr\{\chi_u \ \e^{itH} 
\f(H) \ \e^{\sigma\abs{X_u}^{\zeta}} \ \e^{-itH} \f(H) \chi_u\},
\end{align}
the $(\sigma,\zeta)$-subexponential moment at time $t$ for the time evolution, initially localized near $u \in\z^d$ and localized 
in energy by the smooth function $\f\in\mathcal{C}^{\infty} _{c,+}(I)$.

The following theorem generalizes the main result of  \cite{T}.

 \begin{theorem}\label{equiv1}
Let $I\subset\sigma(H)$ be an interval and assume that $H$ has pure point spectrum in $I$. The following properties are equivalent.
\begin{enumerate}
  \item[(i)] There exist $\sigma>0, \zeta\in(0,1]$ so that for any 
 $\eps>0$, $u\in\z^d$ and  $\f\in\mathcal{C}^{\infty} _{0,+}(I)$, there is a constant $C_{\sigma,\zeta,\eps,\f} <\infty$, 
so that
\begin{equation}\label{exp_DL}
\sup_T \M_u(\sigma,\zeta,\f,T)  :=\sup_T \frac{1}{T}\int_{0}^{T} M_u(\sigma,\zeta,\f,t) \dd t\leq C_{\sigma,\zeta,\eps,\f} \ 
\e^{\eps\abs{u}^{\zeta}}.
\end{equation}
 \item[(ii)] There exist $\sigma>0, \zeta\in(0,1]$ so that for any 
 $\eps>0$, $u\in\z^d$ and  $\f\in\mathcal{C}^{\infty} _{0,+}(I)$, there is a constant $C_{\sigma,\zeta,\eps,\f} <\infty$, 
such that
\begin{equation}
\sup_T \frac1T \int_0^\infty \mathrm{e}^{-t/T} M_u(\sigma,\zeta,\f,t) \dd t \leq C_{\sigma,\zeta,\eps,\f} \ 
\e^{\eps\abs{u}^{\zeta}}.
\end{equation} 
\item[(iii)] There exist $\sigma>0, \zeta\in(0,1]$ so that for any 
 $\eps>0$, $u\in\z^d$ and  $\f\in\mathcal{C}^{\infty} _{0,+}(I)$, there is a constant $C_{\sigma,\zeta,\eps,\f} <\infty$, 
such that
\begin{equation}\label{exp_DL0}
 \sup_t M_u(\sigma,\zeta,\f,t)\leq C_{\sigma,\zeta,\eps,\f} \ \e^{\eps\abs{u}^{\zeta}}.
 \end{equation}
  
\item[(iv)] There exist $\zeta\in (0,1], \sigma>0$ such that for 
all $\eps>0$ and for any $\f\in\mathcal{C}^{\infty} _{0,+}(I)$, there is a constant $C_{\zeta,\sigma,\eps,\f}<\infty$, so that
\begin{equation}\label{SUDL} 
\sup_t \normsch{\chi_x \ \e^{-itH}\f(H) \chi_u}\leq C_{\zeta,\sigma,\eps,\f} \ 
\e^{\eps\abs{u}^{\zeta}}\ \e^{-\sigma\abs{x-u}^{\zeta}} \ \ for \ all \ x,u \in\z^d.
\end{equation}

\item[(v)] There exist $\zeta\in (0,1], \sigma>0$ such that for all 
$\eps>0$ and for any $\f\in\mathcal{C}^{\infty} _{0,+}(I)$, there is a constant $C_{\zeta,\sigma,\eps,\f}<\infty$, so that 
\begin{equation}\label{SULP}
\sup_E \f(E)\normsch{\chi_x P_{E} \chi_u}\leq C_{\zeta,\sigma,\eps,\f} \ \e^{\eps\abs{u}^{\zeta}}
 \  \e^{-\sigma\abs{x-u} ^{\zeta}} \ \ for \ all \ x,u \in\z^d. 
 \end{equation} 
\end{enumerate}
If $H$ satisfies one of these properties, we say that $H$ exhibits  (subexponential) dynamical localization in $I$. When $\zeta=1$ 
we may talk about exponential dynamical localization.

 \end{theorem}

Clearly properties $(iii)$ and $(iv)$ are equivalent, and $(iii)\Longrightarrow (ii) \Longrightarrow (i)$. It will remain to show
 that $(i)\Longrightarrow (v) \Longrightarrow (iv)$, which is the heart of  Theorem~\ref{equiv1}. We point out that the underlying  
geometry of the Hilbert space only plays a role in the proof of $(v) \Longrightarrow (iv)$.

\begin{remark} 
(i) Properties $(iv)$ and $(v)$ of Theorem~\ref{equiv1} have been introduced in \cite{DRJLS2}, and are respectively called by 
Semi-Uniform Dynamical Localization (SUDL)  and Semi-Uniform Localized Projections (SULP).
\\
(ii) Theorem~\ref{equiv1}  actually shows that dynamical localization and time averaged dynamical localization are equivalent. 
\\
(iii) Theorem~\ref{equiv1}  generalizes the result of \cite{T} in the sense that it provides the decay of the kernel of $\e^{-itH}$.
\\
(iv) By the RAGE theorem, the bound \eqref{exp_DL} implies that the spectrum of $H$ is pure point in $I$. If the multiplicity is 
finite, Theorem~\ref{connexion} will show that much more holds true.
\\
(v) If one considers polynomial moments in \eqref{moment} rather than (sub)exponential ones, then Theorem~\ref{equiv1}  still holds 
with polynomial decay in (iv) and (v).
\end{remark}

We turn to the description of the decay properties of the eigenfunctions of $H$ and we start with some notations. 

Let $\E\subset I$ be a collection of eigenvalues of $H$ that we assume to be nonempty ($\E$ may be infinite). 
Set $P_{\E}= \sum_{E\in\E} P_E$ and write $\h_E=P_E\h$ and $\h_{\E}=P_{\E}\h$.  We fix $\kappa > \frac d2$, and define $T$ as the  
operator on $\h$ given by multiplication by 
the function  
$T(x)=\scal{x}^{\kappa}$ for $x\in\R^d$, with $\scal{x}:=\sqrt{1+|x|^2}$. We set
\begin{equation}\label{alhpa E}
\alpha_E: =\tr \{T^{-1}P_E T^{-1}\}= \|T^{-1}P_E\|_2^2\leq \tr P_{E} .
\end{equation}
Given a unit vector  $\phi\in\h$, we denote by $P_\phi$ the rank one projection $P_\phi=|\phi\ra\la\phi|$, and let
\begin{equation}\label{alpha phi}
\alpha_\phi: =\tr \{T^{-1}P_\phi \ T^{-1}\}= \|T^{-1}P_\phi\|_2^2 = \lVert T^{-1}\phi  \rVert^2 \le 1.
\end{equation}
If $\{\phi_n\}_{n=1}^{N_E}$ is an orthonormal basis of $P_E\h$, with
$N_E = \tr P_{E} \leq \infty$,  then $\sum_{n=1}^{N_E} \alpha_{\phi_n}=\alpha_E$.
We assume the following finiteness condition
\begin{equation}\label{sumalphaH}
\alpha_{H,\E} := \sum_{E \in \E} \alpha_E =\tr \{T^{-1}P_{\E} T^{-1}\} <\infty.
\end{equation}
If $\overline{\E}$ is compact,  condition  \eqref{sumalphaH} is known to hold for a large variety of Schr\"odinger and generalized 
Schr\"odinger operators \cite{KKS,GK2}.

\begin{theorem}\label{thmequiv}
 Let $\G_\E=\{\phi_n\}_n$ be an orthonormal basis of $\h_{\E}$, and assume \eqref{sumalphaH}. Then the following properties are equivalent:
\\
(i)  Summable Uniform Decay of Eigenfunction Correlations on $ \G_{\E}$ (SUDEC): there exist $\sigma>0, \zeta\in (0,1]$ 
such that for all $\eps>0$ and all $\phi_n\in\G_{\E}$ and $x,u \in\z^d$, we have
\begin{equation}\label{sudecH}
\normsch{\chi_x P_{\phi_n} \chi_u} =
\norm{\chi_x \phi_n} \norm{\chi_u \phi_n} \leq C_{\zeta,\sigma,\eps} \ \alpha_{\phi_n} \ \e^{\eps \abs{u}^{\zeta}} 
\e^{-\sigma\abs{x-u}^{\zeta}} .
\end{equation}
(i') There exist $\sigma>0, \zeta\in (0,1]$ such that for all $\eps>0$ and all $\phi_n\in\G_{\E}$ and $x,u \in\z^d$,
\begin{equation}\label{sudecH alpha=1}
\norm{\chi_x \phi_n} \norm{\chi_u \phi_n} \leq C_{\zeta,\sigma,\eps} \ \e^{\eps \abs{u}^{\zeta}} 
\e^{-\sigma\abs{x-u}^{\zeta}}.
\end{equation}
(ii) Semi Uniformly Localized Eigenfunctions on $ \G_{\E}$ (SULE): there exist $\sigma>0,
\zeta\in (0,1]$ 
such that for each $\phi_n\in  \G_{\E}$, we can find $x_{\phi_n}\in\z^d$ so that for all $\eps>0$ and $x \in\z^d$, we have
\begin{equation}\label{suleH}
 \norm{\chi_x \phi_n}\leq C_{\zeta,\sigma,\eps} \ \e^{\eps \abs{x_{\phi_n}}^{\zeta}} \e^{-\sigma\abs{x-x_{\phi_n}}^{\zeta}} .
\end{equation}
Moreover, if (ii) holds, we may order the centers of localization $x_{\phi_n}$ in such a way that $| x_{\phi_n} | \ge C n^{1/2\kappa}$.
\end{theorem}

The (SULE) property has been introduced in \cite{DRJLS2}, while the (SUDEC) property has been introduced in \cite{G} and further developed 
in \cite{GK2}. We single out $(i')$ for it may look more natural to the reader. However, while $(i)$ is shown to imply quite readily dynamical 
localization, using $(i')$ would require a more involved analysis.

\begin{remark}\label{sing case}
(i) Notice that if \eqref{sudecH} and \eqref{suleH} are respectively replaced by
\begin{equation}\label{sudec_zeta'}
\norm{\chi_x \phi_n} \norm{\chi_u \phi_n} \leq C_{\zeta',\zeta,\sigma,\eps} \ \alpha_{\phi_n} \ \e^{\eps \abs{u}^{\zeta'}} 
\e^{-\sigma\abs{x-u}^{\zeta}} ,
\end{equation}
and
\begin{equation}\label{sule_zeta'}
 \norm{\chi_x \phi_n}\leq C_{\zeta',\zeta,\sigma,\eps} \ \e^{\eps \abs{x_{\phi_n}}^{\zeta'}} \e^{-\sigma\abs{x-x_{\phi_n}}^{\zeta}} ,
\end{equation}
with $\zeta'>\zeta$ then the equivalence is lost. However (SUDEC), that is \eqref{sudec_zeta'}, is still strong enough to imply dynamical 
localization. This is not the case for  \eqref{sule_zeta'}, because of the lack of the quantity $\alpha_{\phi_n}$. This situation is not as 
exotic as one may think! This is exactly what happens for random Schr\"odinger operators with singular measure (including Bernoulli), 
see \cite{GK3}.
\end{remark}

Until now, the multiplicity of eigenvalues may be arbitrary. Now, we introduce a third class of properties wich corresponds to stronger 
version of (SUDEC) and (SULE) and that will forces multiplicity to be finite.
Our motivation comes from the theory of Anderson localization where eigenfunctions are shown to exhibit stronger localization properties
than (SULE) or (SUDEC). We describe them in the following theorem.

\begin{theorem}\label{thmfinrank}
Assume \eqref{sumalphaH}. Let $E\in\E$ be given ant let $\G_E=\{\phi_n\}_n$ be an orthonormal basis of $\h_{E}$. 
Then the following properties are equivalent:\\
 (i) There exist $\sigma>0, \zeta\in (0,1]$ such that for any $\eps>0$, for all $\phi_n, \phi_m \in\G_E$ 
 and for all $x,u\in\z^d$,
   \begin{equation}\label{sudec+}
\norm{\chi_x \phi_n} \norm{\chi_u \phi_m} \leq C_{\zeta,\sigma,\eps} \sqrt{\alpha_{\phi_n}\alpha_{\phi_m}} \ 
\e^{\eps \abs{u}^{\zeta}} \e^{-\sigma\abs{x-u}^{\zeta}} .
\end{equation}
 (ii) There exist $\sigma>0, \zeta\in (0,1]$ such that for any $\eps>0$, for all $\phi\in\mathrm{Span}\ \G_E$
 and for all $x,u\in\z^d$,
 \begin{equation}
\norm{\chi_x \phi} \norm{\chi_u \phi} \leq C_{\zeta,\sigma,\eps} \ \alpha_{\phi} \ 
\e^{\eps \abs{u}^{\zeta}} \e^{-\sigma\abs{x-u}^{\zeta}} .
\end{equation}
(iii) There exist $\sigma>0, \zeta\in (0,1]$ such that for any $\eps>0$, for all $\phi,\psi\in\mathrm{Span}\ \G_E$ and for all $x,u\in\z^d$,
 \begin{equation}\label{sudec^+}
\norm{\chi_x \phi} \norm{\chi_u \psi} \leq C_{\zeta,\sigma,\eps} \sqrt{\alpha_{\phi} \alpha_{\psi}} \ 
\e^{\eps \abs{u}^{\zeta}} \e^{-\sigma\abs{x-u}^{\zeta}} .
\end{equation}
 (iv) There exist $\sigma>0, \zeta\in (0,1]$ such that for any $\eps>0$, for all $x,u \in\z^d$,
\begin{equation}\label{sudecP}
\normsch{\chi_x P_{E}} \normsch{\chi_u P_{E}} \leq C_{\zeta,\sigma,\eps} \ \alpha_E \ \e^{\eps \abs{u}^{\zeta}} 
\e^{-\sigma\abs{x-u}^{\zeta}} .
\end{equation}
(iv') property (i), (ii), (iii) or (iv) holds with $\alpha_\bullet=1$, $\bullet=\phi, \psi$ normalized vectors or $\bullet=E$, 
as in \eqref{sudecH alpha=1}.
 \\
 (v) There is a common center of localization $x_E$ for all $\phi_n\in \G_\E$
 such that there are $\sigma>0, \zeta\in (0,1]$ so that for any $\eps>0$, for all $x \in \z^d$,
\begin{equation} \label{sule+}
\norm{\chi_{x} \phi_n} \leq C_{\zeta,\sigma,\eps} \sqrt{\alpha_{\phi_n} }\ \e^{\eps|x_E|^\zeta} \e^{-\sigma|x-x_E|^\zeta} .
\end{equation}
(vi) There is a common center of localization $x_E$ for all $\phi\in\mathrm{Span}\ \G_E$ such that 
there are $\sigma>0, \zeta\in (0,1]$ so that for any $\eps>0$, for all $x \in \z^d$,
\begin{equation} \label{sule^+}
\norm{\chi_{x} \phi} \leq C_{\zeta,\sigma,\eps}  \sqrt{\alpha_{E}} \ \e^{\eps|x_E|^\zeta} \e^{-\sigma|x-x_E|^\zeta} .
\end{equation}
(vii) There exists $x_E \in \z^d$ such that there exist $\sigma>0, \zeta\in (0,1]$ so that for any 
$\eps>0$, for all $x \in \z^d$, we have
\begin{equation} \label{suleP}
\normsch{\chi_{x} P_E} \leq C_{\zeta,\sigma,\eps}\sqrt{\alpha_{E}} \ \e^{\eps|x_E|^\zeta} \e^{-\sigma|x-x_E|^\zeta} . 
\end{equation}
We denote by (SUDEC+) any of properties $(i)$, $(ii)$, $(iii)$, $(iv)$, and by (SULE+) any of properties $(v)$, $(vi)$, $(vii)$.
\\
(vii') property (v), (vi) or (vii)  holds with $\alpha_\bullet=1$, $\bullet=\phi, \psi$ normalized vectors or $\bullet=E$.
\\
If one of the above properties holds then the eigenvalues have finite multiplicity and in addition,
\begin{equation}\label{bndtr}
\tr P_{E} \leq C_{\zeta,\sigma} \ \alpha_E  \scal{x_E}^{2\kappa},
\end{equation}
and
\begin{equation} \label{NLH'}
\tilde{N}_L := \# \{E \in\E; |x_E| \leq L \} \leq C_{\zeta,\sigma} \ \alpha_{H,\E}  L^{2\kappa} \quad \text{for all}  \quad   
L \geq 1 ,
\end{equation}
where $x_E$ is as in $(vii)$.
\end{theorem}

\begin{remark}
(i) The bootstrap Multiscale Analysis of \cite{GK1} yields (SULE+) and (SUDEC+). See also \cite{GK2}.
\\
(ii) If \eqref{sule+} holds  with $\eps=0$, then the multiplicity is uniformly finite. This can be seen from Proposition~\ref{centers} , 
since \eqref{center} would hold with $\delta=0$.
\end{remark}

Next, notice that $\normsch{\chi_x P_E \chi_y} \leq  \normsch{\chi_x P_E} \normsch{\chi_y P_E}$, so that \eqref{sudecP} implies 
\eqref{SULP}. One way wonder whether fast decay of $\normsch{\chi_x P_E \chi_y}$ is equivalent to the one of 
$ \normsch{\chi_x P_E} \normsch{\chi_y P_E}$. Such a question was raised in \cite{DRJLS2}. In $\ell^2(\z^d)$, \cite{DRJLS2} proved 
the equivalence when the multiplicity is one ($\tr P_{E}=1$), and  \cite{EGS} showed that if $\tr P_{E} <\infty$ and 
$|\scal{\dirac y, P_E \dirac x}|\leq C_\eps \e^{\eps(|x|+|y|)} \e^{-\sigma|x-y|}$, then 
there exists a basis of $\h_E$ with a (SUDEC) type property.

We summarize the relationships between the three classes of properties in the following optimal theorem. In particular it answers 
to  \cite{DRJLS2}'s  question about the equivalence between (DL) and (SULE), and the role played by the multiplicity (they were considering simple spectrum only).
\begin{theorem}\label{connexion}

Assume \eqref{sumalphaH}. Then \\
(i) We have
 \begin{equation}\label{main impl}\begin{pmatrix} \mathrm{SUDEC +} \\ \mathrm{SULE +} \end{pmatrix} \Longrightarrow 
 \begin{pmatrix} \mathrm{SUDEC} \\ 
\mathrm{SULE}\end{pmatrix} \Longrightarrow
\begin{pmatrix} \mathrm{DL} \end{pmatrix}.
\end{equation}
(ii) Assume that $\tr P_E <\infty$ for any $E\in\E$. Then
 \begin{equation}\label{main equiv}
\begin{pmatrix} \mathrm{SUDEC} \\ 
\mathrm{SULE}\end{pmatrix} \Longleftrightarrow
\begin{pmatrix} \mathrm{DL} \end{pmatrix}.
\end{equation}
(iii) There exist Schr\"odinger operators with eigenvalues of infinite multiplicity and for which (DL) holds but not (SULE/SUDEC). 
There exist Schr\"odinger operators for which (SULE/SUDEC) holds but not (SULE+/SUDEC+), as soon as  eigenvalues are not simple. 
But (SULE/SUDEC) together with property \eqref{center} is equivalent to (SULE+/SUDEC+).
\end{theorem}

\begin{remark}
Of course, when the multiplicity is one, then  (SULE/SUDEC) and (SULE+/SUDEC+) are the same.  (SULE+/SUDEC+) provides a strong condition 
on the spatial repartition of the centers of localization described in Proposition~\ref{centers} below. There is no reason for such a rigid 
condition on centers to hold in great generality, for eigenfunctions associated to a given eigenvalue may live far apart. For instance, 
one may consider the Laplacian on a subgraph of $\z^2$, for which there exist compactly supported eigenfunctions associated to the same 
eigenvalue and with disjoint supports.

However it is easy to see that (SULE/SUDEC) together with the property \eqref{center} implies (SULE+/SUDEC+).
\end{remark}

\medskip

As  previously mentioned, these results remain valid in a general framework
that we briefly outline. Let us consider an abstract separable Hilbert space $\mathbb{H}$ 
equipped with a basis denoted by $\{e_n\}_{n\in\N}$ that we suppose to be orthonormal.
Adopting notations of Section~\ref{results}, we  define the subexponential moment with parameters 
$\sigma$ and $\zeta$:
\begin{equation}\label{sub mom}
M_{e_u}(\sigma,\zeta,\f,t):=\sum_{n\geq 0} \e^{\sigma n^\zeta} |\scal{\e^{-itH}\f(H) e_u,e_n}_{\mathbb{H}}|^2,
\end{equation}
where $e_u\in\mathbb{H}$ is an initial state and $\scal{.,.}_{\mathbb{H}}$ denotes the inner product in $\mathbb{H}$. 
Note that this corresponds to \eqref{moment} with $\chi_u$ replaced by $\Pi_{e_u}$ the rank one projection onto $e_u$.
Theorem~\ref{equiv1} is still valid in this context. Indeed, given $L>0$ and $u\in\N$, we consider the ball $B_L(e_u) := \{e_n, |n-u|\leq L \}$. 
Notice that $\#B_L(e_u)\le 2L+1$ uniformly in $u$, so that Lemma~\ref{order} holds true with $d=1$ in \eqref{N}, which is the only place in 
the proof where the geometry plays a role. 

However, we may object that we loose the physical interpretation of  moments and of dynamical localization. From this point of view it 
is interesting to consider graphs as generalizations of the lattice $\z^d$. Let  $\mathbb{G}$ be a graph with vertices $v\in\mathbb{V}$, 
and set $\mathbb{H}= \ell^2(\mathbb{V})$.
\\
Let $\{\delta_v \}_{v\in \mathbb{V}}$ be the canonical basis of $\ell^2(\mathbb{V})$. We have a natural notion of distance $d$ in 
$\mathbb{V}$: $d(u,v)=\inf \#\{p(u,v)\}$, where $p(u,v)$ is a path in $\mathbb{G}$ joining $u$ and $v$ (if $\mathbb{G}$ is a tree 
then there is only one such path, but $\mathbb{G}$ may contain loops). We can thus define spheres $S_L(u)=\{v\in \mathbb{V}; \ d(u,v)= L\}$ 
centered at $u\in\mathbb{G}$ and of radius $L$. We define $\mathcal{N}_L(u) = \# S_L(u)$.

\begin{theorem}\label{geomthm1}
Assume there exists $\beta\in [0,1)$ such that
\begin{equation}\label{volcnd}
 \sup_u \mathcal{N}_L(u)\leq \e^{L^\beta},
\end{equation}
then Theorem~\ref{equiv1} holds for $\zeta> \beta$.
\end{theorem}

The result thus still applies to graphs but with moderate growth. As example, rooted trees, as in \cite{Br}, satisfy to the growth condition 
\eqref{volcnd}.  And random Schr\"odinger operators on such rooted trees are shown to exhibit dynamical localization  \cite{Br}. See also
\cite{Tau}.

We turn to the (SUDEC) and (SULE) type properties. The geometry is further involved  in the condition
\begin{equation}\label{alphacnd}
 \alpha_{H,\E} = \sum_{u\in\mathbb{G}} \scal{\dirac u,T^{-1} P_\E T^{-1} \ \dirac u} <\infty,
\end{equation}
where $T$ is now the operator given by the multiplication by $e^{|u|^\alpha}$ for fixed $\alpha<\zeta$. We have
\begin{theorem}\label{geomthm2}
Assume \eqref{volcnd} for $\beta\in [0,1)$ and \eqref{alphacnd} holds for $\alpha\in (0,1)$, with $\beta<\alpha<\zeta$,
then Theorem~\ref{thmequiv} and Theorem~\ref{thmfinrank} hold.
\end{theorem}

\section{Proofs}\label{proofs}
\subsection{Dynamical localization}\label{type1}
In this section, we prove Theorem~\ref{equiv1}, as a combination of the theorems below.
\newline 
Given $u\in\z^d$ we consider the function 
\begin{equation}\label{P_u}\p_u(x,\f):=\sup_k \f(E_k)\normsch{\chi_x P_{E_k} \chi_u},\end{equation}
 and its corresponding moment 
\begin{equation}\label{L_u}\style{\lp_u(\sigma,\zeta,\f):=\sum_x \e^{\sigma\abs{x-u}^{\zeta}} \ {\p_u ^2 (x,\f)}},\end{equation}
for $\sigma>0, \zeta\in(0,1]$ and 
where $P_{E_k}$ denotes the eigenprojection associated to the eigenvalue $E_k$. The role of the function $\p_u(x,\f)$ above is to describe the 
decay of the eigenprojectors in terms of the subexponential moment \eqref{moment}, yielding directly \eqref{SULP}.
\begin{theorem}\label{proj}
Fix $\sigma>0$ and $\zeta\in (0,1]$. Then 
\begin{equation}\label{DL/SULP}
\liminf_{T\to\infty} \M_u(\sigma,\zeta,\f,T) \geq C_{\sigma,\zeta} \ \lp_u(\sigma,\zeta,\f).
\end{equation}
for any $\f\in\mathcal{C}^{\infty} _{0,+}(I)$  and all $u \in\z^d$. And thus
\begin{equation}\label{getsulp}
\p_u(x,\f)\leq C_{\sigma,\zeta} \ (\liminf_{T\to\infty} \M_u(\sigma,\zeta,\f,T))^{1/2} \e^{-\frac{\sigma}{2} \ \abs{x-u}^{\zeta}}.
\end{equation}
\end{theorem}
\begin{proof} We first notice that
\begin{equation}\label{lower bnd}
M_u(\sigma,\zeta,\f,T) \geq C_{\sigma,\zeta} \sum_{x\in\z^d} \e^{\sigma\abs{x-u}^{\zeta}} \normsch{\chi_x \ \e^{-itH} \f(H) \chi_u}^2. 
\end{equation}
For $T>0$ and $L\geq 1$ we consider the finite volume time-averaged moment
$$ \M_u^L (\sigma,\zeta,\f,T):=\frac{1}{T}\int_{0}^{T}\sum_{x\in\Lambda_L(u)} \e^{\sigma\abs{x-u}^{\zeta}} \normsch{\chi_x \ 
\e^{-itH} \f(H) \chi_u}^2 \ \dd t. $$
The decomposition of the kernel over the eigenspaces allows us to write
\begin{align}
\M_u^L(\sigma,\zeta,\f,T)=& \sum_{k,k'} \f(E_k)\f(E_{k'}) \sum _{x \in \Lambda_L(u)} \e^{\sigma\abs{x-u}^{\zeta}} \\
& \quad \tr\{\chi_x 
P_{E_k} \chi_u P_{E_k'} \chi_x\}\left(\frac{1}{T}\int_{0}^{T} \e^{-it(E_k -{E_k'})} \ \dd t \right),
\end{align}
and a use of the dominated convergence theorem implies that
$$\lim_{T\to\infty} \M_u^L(\sigma,\zeta,\f,T) = \sum_k \sum _{x \in \Lambda_L(u)} \f^2(E_k) \e^{\sigma\abs{x-u}^{\zeta}} 
\normsch{\chi_x P_{E_k} \chi_u}^2 ,$$
where we have used the fact that
$$\frac{1}{T} \int_{0}^{T} \e^{-it(E_k -E_{k'})} \ \dd t = \begin{cases} 1 & k=k' \\ \frac{\e^{-iT(E_k -E_{k'})}}{-iT(E_k - 
E_{k'})} & k\neq k' \end{cases}.$$
Since $\style{\liminf_{T\to\infty} \M_u(\sigma,\zeta,\f,T) \geq C_{\sigma,\zeta} \lim_{T\to\infty} \M_u^L(\sigma,\zeta,\f,T)}$  
and 
taking the limit when $L\to\infty$, we deduce that
\begin{align}
\liminf_{T\to\infty} \M_u(\sigma,\zeta,\f,T) &\geq C_{\sigma,\zeta} \sum_k \sum _{x\in\z^d} \f^2(E_k) \ \e^{\sigma\abs{x-u}^{\zeta}} 
\normsch{\chi_x P_{E_k} \chi_u}^2 \\
&  \geq C_{\sigma,\zeta} \ \lp_u(\sigma,\zeta,\f).
\end{align} 
As a consequence, \eqref{getsulp} holds.
\end{proof}

 The next statement relates \eqref{SULP} and \eqref{SUDL}.
 \begin{theorem}\label{kernel}
Fix $\sigma>0$ and $\zeta\in (0,1]$ and let $\gamma\in(0,1)$. Then 
 \begin{equation}\label{getsudl}
 \sup_t \normsch{\chi_x \ \e^{-itH} \f(H) \chi_u}\leq C_{\sigma,\zeta,d,\gamma,I} \ \p_u^{1-\gamma}(x,\f) 
\ \lp_u^{\gamma/2}(\sigma,\zeta,\f)
 \end{equation} 
for all $x,u \in\z^d$ and any function $\f\in\mathcal{C}^{\infty} _{0,+}(I)$. 
In particular, 
\begin{equation}
\sup_t \normsch{\chi_x \ \e^{-itH}\f(H) \chi_u} \leq C_{\sigma,\zeta,d,\gamma,I} (\liminf_{T\to\infty} 
\M_u(\sigma,\zeta,\f,T)^{1/2} \e^{-\frac{(1-\gamma)}{2} \ \sigma\abs{x-u}^{\zeta}}. 
\end{equation}
 \end{theorem}

This relies partly on the following lemma which provides a bound on the number of elements contained in a box of size $L$. Its proof is 
given in Appendix A.

\begin{lemma}\label{order} Fix $\sigma>0$ and $\zeta\in (0,1]$. For $k\in\z$ and $u\in\z^d$, we set   
\begin{equation}\label{A_k}
A_k(\sigma,\zeta,u):=\sum_x \e^{\sigma\abs{x-u}^{\zeta}} \frac{\normsch{\chi_x P_{E_k} \chi_u}^2}
{\normsch{\chi_u P_{E_k}} ^2}.
\end{equation} 
 Then 
\begin{equation}\label{N}
 N_{L,\zeta,\sigma,u} :=\sharp \lbrace k\in\z; E_k \in \ I, \ A_k(\sigma,\zeta,u)\leq L\rbrace \leq C_{\sigma,\zeta,d} \ 
(\log L)^{d/{\zeta}}
 \ \ for \ all \ L\in \N,
\end{equation}
where $C_{\sigma,\zeta,u}$ is a positive constant uniform in $u\in\z^d$.\\
In other terms, with new constant we can order $A_k(\sigma,\zeta,u)$ increasingly so that $ A_k(\sigma,\zeta,u)
\geq  \exp(\tilde{C}_{\sigma,\zeta,d} \  k^{{\zeta}/d})$.
 \end{lemma}

\begin{proof}[Proof of Theorem~\ref{kernel}.]
Write
\begin{align}
\sup_t \normsch{\chi_x \e^{-itH} \f(H) \chi_u}
&\leq\sum_{k; E_k \in \ I} \f(E_k)\normsch{\chi_x P_{E_k} \chi_u}\notag\\ 
&\leq \p_u^{1-\gamma}(x,\f) \sum_{k; E_k \in \ I} \f^{\gamma}(E_k)\normsch{\chi_x P_{E_k}}^{\gamma} \normsch{\chi_u P_{E_k}}^{\gamma}. 
\notag 
\end{align}
As in \cite{T}, we shall sacrify some decay in  space in order to recover the summability over $k$. Given $\sigma>0, \zeta\in (0,1]$, one has
$$\f^2(E_k)\normsch{\chi_u P_{E_k}}^2 A_k(\sigma,\zeta,u)=\sum_x \e^{\sigma\abs{x-u}^{\zeta}} 
\f^2(E_k) \normsch{\chi_x P_{E_k} \chi_u}^2\leq \lp_u(\sigma,\zeta,\f).$$
Thus 
\begin{equation}\label{normPk}
\f(E_k)\normsch{\chi_u P_{E_k}}\leq A_k ^{-1/2} (\sigma,\zeta,u) \lp_u ^{1/2}(\sigma,\zeta,\f),
\end{equation}
and 
\begin{equation}\sup_t \normsch{\chi_x \e^{-itH} \f(H) \chi_u}\leq \p_u^{1-\gamma}(x,\f) \lp_u ^{\gamma/2}(\sigma,\zeta,\f) 
\sum_{k; E_k \in \ I}\normsch{\chi_x P_{E_k}}^{\gamma} A_k ^{-\gamma/2} (\sigma,\zeta,u).
\end{equation}

We use need Lemma~\ref{order} to guarantee summability in $k$. Together with H\"{o}lder inequality with conjugate exponents 
$s=2/{\gamma}$ and $s'=\frac{2}{2-\gamma}$, we get
\begin{align}
\sum_{k; E_k \in \ I}\normsch{\chi_x P_{E_k}}^{\gamma} A_k ^{-\gamma/2} (\sigma,\zeta,u)
&\leq \left( \sum_{k; E_k \in \ I} \normsch{\chi_x P_{E_k}}^2\right)^{\gamma/2} \left( \sum_{k; E_k \in \ I}
A_k ^{\frac{-\gamma}{2-\gamma}} (\sigma,\zeta,u)\right) ^{(1-\gamma/2)}\label{geom cnd} \\
&=\ C_{\sigma,\zeta,d,\gamma,I} < \infty.\notag
\end{align}
Hence
\begin{equation}
\sup_t \normsch{\chi_x \ \e^{-itH} \f(H) \chi_u}\leq C_{\sigma,\zeta,d,\gamma,I} \ \p_u ^{1-\gamma}(x,\f)
\lp_u^{\gamma/2}(\sigma,\zeta,\f).
\end{equation}
\end{proof}

\begin{proof}[Proof of Theorem~\ref{equiv1}.] 
\eqref{getsulp} shows that $(i) \Rightarrow (v)$, and \eqref{getsudl} that $(v) \Rightarrow (iv)$. 
 \end{proof}

\subsection{SULE, SUDEC}\label{type2}
We now focus now on the second kind of criteria and we start with 
the proof of Theorem~\ref{thmequiv}. It is a consequence of the 
theorem below which is the main technical result of this section. 
We may omit  the index $n$ and write $\phi\in\G_\E$ instead of $\phi_n\in\G_\E$.

Similary to \eqref{suleH} and \eqref{sudecH}, we shall say that $H$ verifies $(\mathrm{SULE}_f)$/$(\mathrm{SUDEC}_f)$ for some 
function $f$ if these estimates are respectively replaced by 

\begin{equation}\label{suleHf}
 \norm{\chi_x \phi}\leq C_{\zeta,\sigma,\eps} f(\alpha_{\phi}) \ \e^{\eps \abs{x_{\phi}}^{\zeta}} 
\e^{-\sigma\abs{x-x_{\phi}}^{\zeta}} 
\end{equation}
and
\begin{equation}\label{sudecHf}
\norm{\chi_x \phi} \norm{\chi_u \phi} \leq C_{\zeta,\sigma,\eps} \ f(\alpha_{\phi}) \ \e^{\eps \abs{u}^{\zeta}} 
\e^{-\sigma\abs{x-u}^{\zeta}}. 
\end{equation}

\begin{theorem}\label{techthme}
Let $\G_\E=\{\phi_n\}_{n\ge 1}$ be an orthonormal basis of $\h_{\E}$ . Then the following properties are equivalent:
\\
(i) there exists a nonnegative function such that for any $\eps>0$, $f(s)\leq C_{\eps}\ \e^{\eps \ s^{-\zeta/ 2\kappa}}$
 for all $0<s \leq 1$ and for which $H$ has $(\mathrm{SUDEC}_f)$ on  $\G_{\E}$.\\
(ii) there exist $\sigma>0, \zeta\in (0,1]$ such that for any $\epsilon>0$ 
\begin{equation}
\norm{\chi_x \phi} \norm{\chi_u \phi} \leq C_{\sigma,\zeta,\eps} \ \e^{\eps \abs{u}^{\zeta}} \e^{-\sigma\abs{x-u}^{\zeta}},
\end{equation}
for all $\phi\in \G_{\E}$ and all $x,u \in \z^d$.\\
(iii) $H$ exhibits (SUDEC) on $\G_{\E}$.\\
(iv) For any nonnegative function such that for any $\eps>0$, $f(s)\geq C_{\eps} \ \e^{-\eps \ s^{-\zeta/ 2\kappa}}$ for all $0<s
 \leq 1$, $H$ has $(\mathrm{SUDEC}_f)$ on  $\G_{\E}$.
\end{theorem}

Recall $\alpha_\phi\le 1$. Obviously, $(iv)\Rightarrow(iii)\Rightarrow(ii)\Rightarrow(i)$. It remains to prove that $(i)\Rightarrow(iv)$. 
This will be a consequence of the next two lemmas.

 \begin{lemma}\label{lemf}
Let $f:\R^+ \to\R^+$ be a function. If there exist $\zeta\in (0,1]$ and $\sigma>0$ such that for all $\eps>0$,
\begin{equation}\label{sudecf}
\norm{\chi_x\phi}\norm{\chi_u\phi} \leq C_{\zeta,\sigma,\eps} f(\alpha_{\phi}) \ \e^{\eps|u|^\zeta}\  \e^{-\sigma|x-y|^\zeta} ,
\end{equation}
for all $x,u\in\z^d$ and any $\phi\in\G_\E$
then there is a new constant $C_{\zeta,\sigma,\eps}$ so that for all $x\in\z^d$, we have
\begin{equation}\label{sulef}
\norm{\chi_x\phi}\leq C_{\zeta,\sigma,\eps}
\frac1{\sqrt{\alpha_{\phi}}}f(\alpha_{\phi}) \ \e^{\eps|x_{\phi}|^\zeta}
\  \e^{-\sigma|x-x_{\phi}|^\zeta},
\end{equation}
where $x_\phi$ maximizes $x\mapsto\norm{\chi_x \phi}$.

In particular, taking $f(s)=s$ says that if (SUDEC) holds on $\G_\E$ then (SULE) holds on $\G_\E$ and with the same parameters 
$\zeta$ and $\sigma$.

\end{lemma}
This lemma tells us that if $\mathrm{(SUDEC)}_f$ holds for a given function $f$ then $\mathrm{(SULE)}_g$ occurs where 
$g:s \mapsto \frac{f(s)}{\sqrt s}$.
\begin{proof}
We set $\tilde{\phi} = \frac1{\sqrt{\alpha_\phi}}\phi = \phi / \|T^{-1}\phi\|$ and we pick $x_\phi \in \z^d$ (not unique) such that 
\begin{equation}\|\chi_{x_\phi}\tilde{\phi}\|=\max_{u \in \z^d} \|\chi_u\tilde{\phi}\|.\end{equation}
Since
\begin{align}
1=\|T^{-1} \tilde{\phi}\|^2 = 
\sum_{u \in \z^d} \|\chi_u T^{-1}\tilde{\phi}\|^2 
\leq  \|\chi_{x_\phi}\tilde{\phi}\|^2 \sum_{u\in\z^d} \norm{\chi_u T^{-1}}^2 \leq C_d \|\chi_{x_\phi}\tilde{\phi}\|^2 ,
\end{align}
 we get
\begin{equation}
\|\chi_{x_\phi}\tilde{\phi}\| \geq C_d^{-1/2}.
\end{equation}
It follows now from \eqref{sudecf} that 

\begin{align}
\|\chi_x{\phi}\|
& \leq  C_d^{1/2}\frac{1}{\sqrt\alpha_\phi}  \|\chi_x{\phi}\| \|\chi_{x_\phi}{\phi}\| \notag \\
& \leq  C_{d,\zeta,\sigma,\epsilon} \frac{f(\alpha_\phi)}{\sqrt\alpha_\phi} \ \e^{\epsilon|x_\phi|^\zeta}
  \e^{-\sigma|x-x_\phi|^\zeta},
\end{align}
for all $x\in\z^d$.
\end{proof}
Furthermore, we establish a control on $\alpha_{\phi}$ in term of the center of localization $x_{\phi}$ according to:
\begin{lemma}\label{lemalphax}
Suppose that $(\mathrm{SULE}_f)$ holds with some function $f:\R^+ \to\R^+$
such that for all $\eps>0$,
\begin{equation} f(s)\leq C_{\eps} \ \e^{\eps s^{-\zeta/2\kappa}} \mathrm{for \ all} \ s \in ]0,1].
\end{equation}
Then there exists a constant $C>0$ (independant of $\G_{\E}$), so that
\begin{align}\label{bnd alpha phi}
 \alpha_{\phi} \geq C \scal{x_{\phi}}^{-2\kappa}  \quad \text{for all $\phi\in\G_{\E}$}. 
\end{align}

\end{lemma}

\begin{proof}
We note that from \eqref{suleHf} we get  
\begin{align}\label{1/9}
\norm{\chi_{|x-x_{\phi}|\geq R} \ \phi}^2 \leq C_{\zeta,\sigma,\eps}^2 f^2(\alpha_{\phi}) \ \e^{2\eps|x_{\phi}|^\zeta} 
\sum_{|x-x_{\phi}|\geq R} \e^{-2\sigma|x-x_{\phi}|^\zeta} \leq \frac{1}{9} ,
\end{align}
if we take 
\begin{equation}\label{Rphi}
R\geq R_{\phi}:=(\frac{\eps}{\sigma})^{1/\zeta}|x_{\phi}| + (\frac{1}{\sigma}\log f(\alpha_{\phi})+ \frac{1}{\sigma}\log(3 
\ C_{\zeta,\sigma,\eps}))^{1 / \zeta}
\end{equation}
Since $|x-x_{\phi}|\leq R_{\phi}$ implies that $|x|\le |x_{\phi}| + R_{\phi}$ and using \eqref{1/9} and \eqref{Rphi}, we have
\begin{align}
\alpha_{\phi}=\norm{T^{-1} \phi}^2
&\geq \sum_{x\in\Lambda_{R_\phi}(x_{\phi})} \norm{\chi_x T^{-1} \phi}^2 \geq 
 \scal{|x_{\phi}|+R_{\phi}}^{-2\kappa} \norm{\chi_{\Lambda_{R_\phi}(x_{\phi})} \phi}^2\notag\\
&\geq  \frac89  \left\{(1+(\frac{\eps}{\sigma})^{1/\zeta})|x_{\phi}| + (\frac{\eps}{\sigma})^{1/\zeta} {\alpha_{\phi}}^{-1/{2\kappa}}
+ C^\prime_{\zeta,\sigma,\eps}\right\}^{-2\kappa}, \notag
\end{align}
for any $\phi\in\G_{\E}$. Thus, choosing $\eps$ small enough, yields \eqref{bnd alpha phi} .
\end{proof}

We complete the proof of Theorem~\ref{techthme}. 

\begin{proof}[Proof of Theorem~\ref{techthme}.]
As mentioned above, it is enough to prove that $(i)$ implies $(iv)$. If there exists a 
function $f$ such that for any $\eps>0$, we have
$$f(s)\leq C_{\eps} \e^{\eps s^{-\zeta/2\kappa}} \quad{for \ all} \ s \in ]0,1],
$$ 
and \eqref{sudecHf} holds, then the (SULE) property \eqref{suleHf} will occur with a factor $\frac{f(\alpha_\phi)}{\sqrt{\alpha_\phi}}$
in view of Lemma~\ref{lemf}. Proceeding now as in \cite[Proof of Proposition~A.1]{G} and making use of \eqref{suleHf}, we get

\begin{align}
\norm{\chi_x \phi} \norm{\chi_{u} \phi} 
& \leq 
\frac1{{\alpha_\phi}} \ f(\alpha_\phi)^2 C^2_{\zeta\sigma,\eps} \e^{2\eps|x_\phi|^\zeta} 
\e^{-\sigma|x-x_\phi|^\zeta-\sigma|u-x_\phi|^\zeta} 
\\
&\leq 
\frac1{{\alpha_\phi}} \ f(\alpha_\phi)^2 C^2_{\zeta,\sigma,\eps} \e^{(2\eps-\eps')|x_\phi|^\zeta} 
\e^{\eps'|u|^\zeta} \e^{-(\sigma-\eps')|x-u|^\zeta},
\end{align}
for all $x,u\in\z^d$ and with $\eps'<\sigma$. We note that it follows from \eqref{bnd alpha phi} that 
\begin{align}
\frac1{\alpha_\phi} \ f(\alpha_\phi)^2
&\leq
\e^{-\eps(\alpha_\phi)^{-\zeta/2\kappa}}
C\scal{x_\phi}^{2\kappa} \e^{3\eps(C\scal{x_\phi})^{\zeta}} \\
& \leq 
\e^{-\eps(\alpha_\phi)^{-\zeta/2\kappa}}
\e^{C_1\eps|x_\phi|^\zeta+C_2},
\end{align}
for some postive and finite constants $C_1, C_2$.
Taking $\eps'> (C_1+2)\eps$, we conclude that \eqref{sudecHf} follows for any function $f\geq0$ such that for any $\eps>0$,
 $f(s)\geq C_{\eps} \e^{-\eps \ s^{-\zeta/ 2\kappa}}$ for all $0<s \leq 1$.
\end{proof}

\begin{proof}[Proof of Theorem~\ref{thmequiv}.]
 The``equivalence" part of the  proof is currently provided by Theorem~\ref{techthme} and Lemma~\ref{lemf}. It remains to show that
the centers of localization $\{x_{\phi_n}\}_n$ can be reordered in such a way that $|x_{\phi_n}|$ increases with $n$. We proceed as in  
\cite{DRJLS2}. 

Given $L>0$, let $R_L := \delta L+ C_\delta$ as in \eqref{Rphi} for some $\delta>0$ (that depends on $\zeta$ and $\sigma$) and where we 
have taken $f\equiv1$, it follows from 
\eqref{1/9} that 
\begin{equation}
\|\chi_{x_{\phi_n}, R_L} \phi_n\|^2 > \frac{1}{9} \ \ \mathrm{whenever} \ \ |x_{\phi_n}|\leq L ,
\end{equation}
and if $N_L$ is the cardinal of the set  $ \{n, \phi_n\in\G_\E; |x_{\phi_n}| \leq L \} $
then we conclude that
\begin{align}
\frac19 N_L 
& \leq \sum_{n, |x_{\phi_n}|\le L }  \|\chi_{x_{\phi_n}, R_L} \phi_n\|^2 
\le \normsch{\chi_{0, L+R_L} P_\E}^2 \notag\\
& \leq C L^{2\kappa} \alpha_{\E},\label{NLH}
\end{align}
for some finite constant $C$ that depending $\zeta$ and $\sigma$.
Since $N_L <\infty$ for all $L>0$ by, \eqref{sumalphaH}, we may reorder the centers of localization in increasing order in $n$, which yields
$|x_{\phi_n}|\geq C _{\E,\sigma,\zeta} \ n^{\frac{1}{2\kappa}}$. 

\end{proof}

We now turn to Theorem~\ref{thmfinrank} and the strong forms of (SUDEC) and (SULE).

\begin{proposition}\label{centers}
Assume (SULE)/(SUDEC) for vectors in the range of $P_E$. For any $\delta>0$ there is a constant $C_\delta$ such that, 
for any $\phi,\psi\in\mathrm{Ran} \ P_E$ and $E\in\E$,   their localization centers $x_\phi , x_\psi$ satisfy
\begin{equation}\label{center}
 |x_\phi - x_{\psi}|\leq \delta |x_\phi|+C_\delta.
\end{equation}
\end{proposition}

\begin{proof} Without loss of generality, we may suppose that $\phi,\psi$ are orthonormalized. We mainly use \eqref{Rphi} 
where we take $f\equiv1$ that yields that for any $\delta>0$,
$$
\norm{\chi_{|x-x_\phi|\ge R_\phi} \phi} \leq \frac13 \ \quad{\mathrm{for}} \ R_\phi=\delta|x_\phi| + C_\delta.
$$ 
If $|x_\phi-x_\psi|\le 2(R_\phi + R_\psi)$, then \eqref{center} follows from the definition of $R_\phi, R_\psi$. 
Assume $|x_\phi-x_\psi|\ge 2(R_\phi + R_\psi)$ and set
$\varphi=\frac1{\sqrt{2}}(\phi+\psi)\in \mathrm{Ran} P_E$. As a consequence,
\begin{align}
\norm{\chi_{|x-x_\phi|\le R_\phi} \varphi}
&\ge \frac1{\sqrt{2}} \norm{\chi_{|x-x_\phi|\le R_\phi} \phi} - \frac1{\sqrt{2}} \norm{\chi_{|x-x_\phi|\le R_\phi} \psi} \notag\\
&\ge \frac2{3\sqrt{2}} - \frac1{\sqrt{2}} \norm{\chi_{|x-x_\psi|\ge R_\psi} \psi} \notag\\
&\ge \frac2{3\sqrt{2}} - \frac1{3\sqrt{2}} = \frac1{3\sqrt{2}}.
\end{align}
In the same manner, we have $\norm{\chi_{|x-x_\psi|\le R_\psi} \varphi} \ge\frac1{3\sqrt{2}}$. 
Having in mind that we assumed 
$|x_\phi-x_\psi| - (R_\phi + R_\psi) \ge \frac12{|x_\phi-x_\psi|}$ and applying (SUDEC) to $\varphi$, we get
\begin{align}
\frac1{18} &\le
\norm{\chi_{|x-x_\phi|\le R_\phi} \varphi}\norm{\chi_{|x-x_\psi|\le R_\psi} \varphi} \notag\\
& \le C_{\zeta,\sigma,\eps} \  \e^{C^{'}_{\zeta,\delta} \eps|x_\phi|^\zeta} 
\e^{-\sigma(|x_\phi-x_\psi|-(R_\phi +R_\psi))^\zeta}\notag\\
& \le C_{\zeta,\sigma,\eps} \ \e^{C^{'}_{\zeta,\delta} \eps|x_\phi|^\zeta} \e^{-\sigma(\frac12|x_\phi-x_\psi|)^\zeta}.
\end{align}
The result follows.
\end{proof}

\begin{remark}
 Notice that \eqref{center} asserts that if (SULE) holds for all vectors in the span of $\G_\E$ then the 
multiplicity has to be finite, since a ball of given radius can only contain a finite number of centers of localization by 
Theorem~\ref{thmequiv}.
\end{remark}

\begin{proof}[Proof of Theorem~\ref{thmfinrank}.] 
 
Since $\norm{\chi_x \phi}\leq \normsch{\chi_x P_E}$ 
for any $\phi\in \mathrm{Ran} P_E$, we immediately get $(iv)\Rightarrow(iii)\Rightarrow(ii),(i)$,
and $(vii)\Rightarrow(vi)\Rightarrow(v)$. Next, we have $(iv)\Leftrightarrow(vii)$ using the same strategy as in the proof of 
Theorem~\ref{thmequiv}.

To see that $(i) \Rightarrow(iv)$, let $(\phi_n)_{n\geq 1}$ be an orthonormalized 
basis of $\mathrm{Ran}\ P_E$ verifying \eqref{sudec+}. Then 
\begin{align}
\normsch{\chi_x P_E}^2 \normsch{\chi_u P_E}^2 
& = \sum_{n,m} \norm{\chi_x \phi_n}^2 \norm{\chi_u \phi_m}^2 \notag\\
& \leq \left(\sum_n \alpha_{\phi_n}\right)^2 {C_{\zeta,\sigma,\eps}^2} \ \e^{2\eps(|x|^\zeta+|u|^{\zeta})} \e^{-2(\sigma-\eps) |x-u|^{\zeta}}
\notag\\
&= {C_{\zeta,\sigma,\eps}^2} \ {\alpha_E}^2 \ \e^{2\eps(|x|^\zeta+|u|^{\zeta})} \e^{-2(\sigma-\eps) |x-u|^{\zeta}}. \label{Pux}
\end{align}
Finite multiplicity follows. Indeed, there exists $u\in\z^d$ such that $\normsch{\chi_u P_E}\neq0$ (otherwise $\tr P_E=0$),  hence 
for all $E\in\E$, $\tr P_E =\style{\sum_{x\in\z^d} \normsch{\chi_x P_E}^2}<\infty$ by \eqref{Pux}.
Next, we show that $(v)\Rightarrow (i)$. We write
\begin{align}
 \norm{\chi_x \phi_n} \norm{\chi_x \phi_m}
&\leq C_{\zeta,\sigma,\eps}^2 \ \e^{2\eps|x_E|^{\zeta}} \e^{-\sigma (|x-x_E|^{\zeta}+|u-x_E|^{\zeta})}\notag\\
&\leq C_{\zeta,\sigma,\eps}^2 \ \e^{-2\eps|x_E|^{\zeta}} \e^{2\eps(|x|^\zeta+|u|^{\zeta})} \e^{-(\sigma-2\eps) |x-u|^{\zeta}},\notag
\end{align}
with $\eps<\sigma/2$. Then $(i)$ follows since $\e^{-2\eps|x_E|^{\zeta}}  \le  C\scal{x_E}^{-2\kappa}\le \sqrt{\alpha_{\phi_n}\alpha_{\phi_m}}$ 
by  \eqref{bnd alpha phi}.

We thus have $(iv) \Rightarrow (vii) \Rightarrow  (v) \Rightarrow  (i) \Rightarrow  (iv)$, and the equivalence is proved 
($(iv')$ and $(vii')$ can be deduced from Theorem~\ref{thmequiv} and Lemma~\ref{techthme}).  At last, we show that $(ii)\Rightarrow (vi)$. 
We have to show that we can get (SULE) with a common center of localization. By Lemma~\ref{lemf} we get a (SULE) bound for all $\phi\in\G_\E$, 
with centers of localization $x_\phi$. Let $x_\psi$ be one of them, but given. By Proposition~\ref{centers}, 
$|x_\phi-x_\psi|\le \delta |x_\psi|+C_\delta$, and 
\begin{align}
 \norm{\chi_x \phi}
 & \leq C_{\zeta,\sigma,\eps} \e^{\eps \abs{x_{\phi}}^{\zeta}} \e^{-\sigma\abs{x-x_{\psi}}^{\zeta}+\sigma\abs{x_\phi-x_{\psi}}^{\zeta}} \\
 & \leq C_{\zeta,\sigma,\eps,\delta} \e^{\eps \abs{x_{\phi}}^{\zeta}+\sigma\delta^\zeta|x_\psi|^\zeta} \e^{-\sigma\abs{x-x_{\psi}}^{\zeta}}
  \\
 & \leq C_{\zeta,\sigma,\eps,\delta} \e^{\eps' \abs{x_{\psi}}^{\zeta}} \e^{-\sigma\abs{x-x_{\psi}}^{\zeta}} ,
\end{align}
with $\eps'=\eps(1+\delta)^\zeta + \sigma\delta^\zeta$.

The bound \eqref{bndtr} is given by an argument similar to the proof of Lemma~\ref{lemalphax}. Indeed,
there are $\zeta\in(0,1], \sigma>0$ such that for any $\eps>0$ there is 
a finite constant $C'_{\zeta,\sigma,\eps}$ for which
\begin{align}
\norm{\chi_{|x-x_E|\geq R_E} P_E}^2 \leq \frac12, \mbox{ where }  R_E=(\frac{\eps}{\sigma})^{1/\zeta}|x_E|+
C'_{\zeta,\sigma,\eps}. \label{P_E}
\end{align}
Since
\begin{align}
\norm{\chi_{|x-x_E|\leq R_E} P_E}^2 \leq \norm{\chi_{|x|\leq (1+\frac{\eps}{\sigma})^{1/\zeta}|x_E|+
C'_{\zeta,\sigma,\eps} } P_E}^2, 
\end{align}
and with $\eps$ small enough one gets 
\begin{align}
\norm{\chi_{|x-x_E|\leq R_E} P_E}^2 \leq C_{\zeta,\sigma} \scal{x_E}^{2\kappa} \alpha_E,
\end{align}
and thus $\tr{P_E}=\norm{P_E}_1 = \normsch{P_E}^2 \leq \frac12 + C_{\zeta,\sigma} \scal{x_E}^{2\kappa} \alpha_E $.
Finally, the last bound \eqref{NLH'} could be deduced from the equation \eqref{P_E} and in proceeding analogously to \eqref{NLH}.

\end{proof}

\begin{proof}[Proof of Theorem~\ref{connexion}.]
The first claim follows immediately from \eqref{sudec+} applied to the case $n=m$ and from \eqref{sudecP} that we combine with
$\normsch{\chi_x P_E \chi_u}\leq \normsch{\chi_x P_E} \normsch{\chi_u P_E}$.
For the second part, notice that the implications from the left to the right are still valid. The novelty here is that under
the hypothesis of finite multplicity, all these properties become equivalent. \newline
Assuming that $H$ exhibits \eqref{SULP} in $\overline\E$, 
we construct a family $\G_\E$ of orthonormalized eigenfunctions that verifies \eqref{suleH}, namely (SULE) property. For any given $E\in \E$, 
since ${\sum_{x\in\z^d}} \normsch{\chi_x P_E}^2 =\tr{P_E}=N<\infty$ there exists $ x_E \in\z^d$ which maximizes 
$\normsch{\chi_x P_E}$. Note that $ \normsch{P_E\chi_{x_E}}\neq 0$, otherwise we would have $\normsch{P_E\chi_x}=0$ for all $x$ 
which is not possible since $\tr{P_E}\neq 0$. Now, we pick a unit vector $\eta\in\h$ such that $\norm{\eta}=1$ and  
$\norm{P_E\chi_{x_E}\eta} \geq  \frac12 \norm{P_E\chi_{x_E}}$, and set
\begin{align}
\phi_1 = \frac{P_E \chi_{x_E} \eta}{\norm{P_E \chi_{x_E}\eta}} \in P_E\h=\h_E.
\end{align}
We have
\begin{align}
\alpha_1 &:= \tr (T^{-1} P_{\phi_1} T^{-1}) = \norm{ T^{-1}\phi_1}^2 \\
& = \sum_{x\in\z^d} \norm{\chi_x T^{-1}  \phi_1}^2 \notag\\ 
& \leq \sum_{x\in\z^d}\norm{\chi_x T^{-1}}^2  \ \frac{\norm{\chi_x P_E \chi_{x_E} \eta}^2}{\norm{P_E \chi_{x_E}\eta}^2} \leq 
\sum_{x\in\z^d} \norm{\chi_x T^{-1}}^2  \norm{\chi_{x} P_E }^2 \notag\\
& \leq C_d  \norm{P_E  \chi_{x_E}}^2 
 \leq 4 C_d \norm{P_E  \chi_{x_E} \eta}^2. \label{alpha1}
\end{align}
As $$\norm{\chi_x \phi_1}\leq \frac{\normsch{\chi_x P_E \chi_{x_E}}}{\norm{P_E \chi_{x_E} \eta}} , $$
we get from \eqref{SULP} and \eqref{alpha1}, that
\begin{align}\label{sqrtalpha1} 
\norm{\chi_{x} \phi_1} \leq \tilde{C}_{\zeta,\sigma,\epsilon} \frac1{\sqrt{\alpha_1}} \ \e^{\epsilon |x_E|^\zeta} 
\e^{-\sigma|x-x_E|^\zeta}.
\end{align}
We repeat this procedure with $P_{E,1}:= P_E - P_{\phi_1}$, and so on with $P_{E,n+1} := P_{E,n} - P_{\phi_{n+1}}$, until the 
rank is zero. The finiteness of the rank of $P_E$, denoting by $N$, ensures that the process will stop. 
Notice that the projectors $P_{E,n}$ exhibit \eqref{SULP}. For instance $\normsch{\chi_x P_{E,1} \chi_u}$ is a sum of two decaying quantities 
\begin{align}
 \normsch{\chi_x P_{E,1} \chi_u}
&\leq  \normsch{\chi_x P_E \chi_u} +\norm{\chi_x \phi_1} \norm{\chi_u \phi_1} \label{sudec_phi1}\\
&\leq C_{\zeta,\sigma,\eps} \ \e^{\eps|u|^\zeta} \e^{-\sigma|x-u|^\zeta},\notag
\end{align}
where the decay of the second term in the r.h.s of \eqref{sudec_phi1} results from \eqref{sqrtalpha1}.
Therefore, by induction 
we get $N$ orthonormalized functions $\phi_n$ satisfying the (SULE)-like estimate in the sense that 
\begin{align}\label{sqrtalpha_n} 
\norm{\chi_{x} \phi_n} \leq \tilde{C}_{\zeta,\sigma,\eps} \ \frac1{\sqrt{\alpha_n}} \ \e^{\eps |x_{E_n}|^\zeta} 
\e^{-\sigma|x-x_{E_n}|^\zeta},
\end{align}
for any $n\in\{1,\dots, N\}$.
We can get rid of $\alpha_n^{-1/2}$ from the proof of Theorem~\ref{techthme}, in which case only an arbitrary small fraction 
of the mass $\sigma$ is lost. Alternatively, at each step, one can follow \cite[Proof of Lemma 4]{EGS} and bound 
$\norm{\chi_{x} \phi_n}$ by the geometric mean of \eqref{SULP} and $\norm{\chi_{x} \phi_n} \leq \norm{\chi_{x_{E_n}} P_{E,n}}$. 
In this latter case, the final $\sigma$ is divided by $2$ at each step.

\end{proof}

We turn to the proof of Theorems \ref{geomthm1}  and  \ref{geomthm2} .
Theorem~\ref{geomthm1} follows immediately from the proof of Theorem~\ref{equiv1} . The main point is to notice that the technical 
Lemma~\ref{order} is still valid in the case of subexponential growth, where the r.h.s of \eqref{N} becomes 
$\e^{C_{\sigma,\zeta,\beta}(\log L)^{\beta/\zeta}}$.\newline
In view of the proof Theorem~\ref{thmfinrank}, the Theorem~\ref{geomthm2} can be deduced by adapting the different steps which involve
the geometry of the space. In particular, the technical result in Lemma~\ref{lemalphax} and Theorem~\ref{techthme} remain true if we take 
$f(s) \le C_\eps \e^{ (- \eps \log s)^{\zeta / \alpha}}$ in $(i)$ and $f(s) \ge C_\eps \e^{- (- \eps \log s)^{\zeta / \alpha}}$ in $(iv)$
for $s\in (0,1]$ in which case $\alpha_\phi\ge C \e^{-|x_\phi|^\alpha}$.

\section{Counterexamples}\label{counterex}
 
The first model is the free Landau Hamiltonian $H_B:= (-i\nabla-A)^2$ on $L^2(\R^2)$ where $A$ is the vector potential 
$A=\frac{B}{2} (-x_2,x_1)$ and $B>0$ is the strength of the constant magnetic field. It is well known that the Landau levels are 
infinitly degenerated and that it exhibits the property \eqref{SULP} and thus dynamical localization. 
We claim that (SUDEC) does not occur for $H_B$. In fact, consider 
for instance the eigenfunctions associated to the first Laudau level and whose expression is given by
\begin{equation}
 \varphi_n (z) = \left(\frac{B^n}{2\pi \ 2^n \ n!} \right)^{1/2} z^n \ \e^{-\frac{B}{4} |z|^2}.
\end{equation}
For $n$ integer, we define the radial function $f_n(r)= r^{2n} \ \e^{-\frac{B}{2} r^2}$ for which the maximum is achieved for the 
radius $ r_{\mathrm{max}}= (\frac{2n}{B})^{1/2}$. Let $z_1$ and $z_2$ to be affixes of two opposite points on this maximal circle. A
simple computation yields

\begin{equation}
 |\varphi_n (z_1) \varphi_n (z_2)| =\frac{n^n}{2\pi n!} \ \e^{-n}.
\end{equation}
Together with the Stirling's formula, it gives that there are no positive constants $c_1$ and $c_2$ ($c_2$ depends on $B$) 
such that $\frac{1}{\sqrt{n}} \leq c_1 \e^{-c_2\sqrt n}$ for all $n$.
\begin{remark}
Another way to see that $H_B$ does not has (SUDEC) can be derived from the theory of the quantum Hall effect.
Indeed, if (SUDEC) would occur for a basis of eigenvectors then the Hall conductance $\sigma_H$ would be constant at Landau levels by 
\cite{GKS1}, while $\sigma_H$ is known to have jumps.

\end{remark}

Next, let us consider the discrete Laplacian $-\Delta$ on subgraphs of $\z^2$. It is enough to consider a subgraph given
by $J\geq2$ disjoint copies $\cc_j $ of a given finite cluster $\cc_1$ and we set $H_j :=-\Delta_{\mid_{\cc_j}}$.  The operators $H_j$, $j=1, \cdots, J$, have the same discrete spectrum with compactly supported eigenfunctions.
The operator $-\Delta_{\mid{\cup_j \cc_j}}=\oplus_j H_j $ for $1\leq j\leq J$ for $1\leq j\leq J$,
has (SULE) since we obtain a basis of compactly supported eigenfunctions. But (SULE+) and (SUDEC+) does not hold as soon as copies $\cc_i$
and $\cc_j$ for $i\neq j$, are far enough so that Proposition~\ref{centers} is violated. 
We mention that such finite clusters appear in a natural way in percolation theory. We refer to \cite{KiM} for more details.

\section{Appendix A}\label{A}
In this section, we shall order the moments \eqref{A_k} given in Section~\ref{proofs}, eq \eqref{A_k}, uniformly on the space.
\begin{proof}[Proof of Lemma~\ref{order}.]
Set $$a_{kx}(u):=\frac{\normsch{\chi_x P_k \chi_u}^2}{\normsch{\chi_u P_k} ^2},$$ which verify
\begin{equation}\label{cnd1}
\sum_x a_{kx}(u)=\frac{\normsch{\chi_u P_k} ^2}{\normsch{\chi_u P_k} ^2} =1 \ \ \ \mathrm{for \ all} \ u\in\z^d \ \mathrm{and \ all} \ k\in\z,
\end{equation}
and 
\begin{equation}\label{cnd2}
\sum_{k, \ E_k \in \ I} a_{kx}(u) \leq \sum_{k, \ E_k \in \ I} \normsch{\chi_x P_k} ^2 = \tr(\chi_x P_I \chi_x)\leq 1 \ \ \
\forall x,u\in\z^d,\end{equation}
where $P_I$ denotes the projection on the interval $I$.\\
For $L\in\N$, define the following set 
$$J_u(L):=\{k\in\z, E_k \in \ I ; \sum_{x\notin \ \Lambda_L(u)} a_{kx}(u)\leq 1/2\},$$ 
and consider the sum
$$S_u(L):=\sum_{k\in J_u(L)}\sum_{x\in \ \Lambda_L(u)} a_{kx} (u).$$
We will estimate the cardinal of $J_u (L)$ in term of the volume of the box $\Lambda_L (u)$.
Note that it follows from \eqref{cnd1} that for $ k\in J_u(L)$ , we have 
$$\sum_{x\in \ \Lambda_L(u)} a_{kx}(u)=\sum_x a_{kx}(u)-\sum_{x\notin \ \Lambda_L(u)} a_{kx}(u)\geq 1/2.$$
Thus \ $ S_u(L)\geq \frac{1}{2} \sharp(J_u(L))$. Moreover, the bound \eqref{cnd2} yields
$$\style{S_u(L)\leq\sum_{k, \ E_k \in \ I} \sum_{x\in \ \Lambda_L(u)} a_{kx}(u)\leq \sum_{x\in \ \Lambda_L(u)} 1 \leq C_d L^d},$$
and hence 
\begin{equation}\label{card J_u}\sharp(J_u(L))\leq C_d L^d.\end{equation}
Now given $\sigma>0$ and $\zeta\in (0,1]$, we set 
$$I_u(L,\sigma,\zeta)=\{k\in\z, E_k \in \ I ; \ A_k(\sigma,\zeta,u)\leq \frac{1}{2} \ \e^{\sigma L^{\zeta}}\},$$  and notice that 
$$\style{A_k(\sigma,\zeta,u)\geq \e^{\sigma L^{\zeta}} \sum_{x\notin \ \Lambda_L(u)} a_{kx}(u)},$$
which shows that \ $I_u(L,\sigma,\zeta)\subset J_u(L)$. Taking the exponential rescaling $l=\frac{\e^{\sigma L^{\zeta}}}{2}$ and
using \eqref{card J_u}, we obtain
$$N(l):=\sharp\{k\in\z, E_k \in \ I ; \ A_k(\sigma,\zeta,u)\leq l \} 
\leq C_{\sigma,\zeta,d} \  (\log l)^{d/{\zeta}}, $$
and thus the finitness of the set $\{k\in\z, \ A_k(\sigma,\zeta,u)\leq l \}$ follows.\\ 
For any $u \in\z^d$, there exists a new order $j_u: k\mapsto j_u(k)$ for $k\in\z$ in such a way that
$A_{j_u(k)}(\sigma,\zeta,u)$ increases. So
$N(A_{j_u(k)})=\abs{j_{u(k)}}$ and with $A_{j_{u(k)}}(\sigma,\zeta,u)=l$, one gets
 $$\abs{j_u(k)}\leq C_{\sigma,\zeta,d} \ \left( \log (A_{j_u(k)}(\sigma,\zeta,u))\right) ^{d/{\zeta}}.$$
We conclude that $A_k (\sigma,\zeta,u)$ may be ordered so that the increase with $k$ in the sense
$$ A_k (\sigma,\zeta,u)\geq \e^{\tilde{C}_{\sigma,\zeta,d} \ k^{{\zeta}/d}},$$
for a positive constant $\tilde{C}_{\sigma,\zeta,d}$ which is uniform in $u\in\z^d$.
\end{proof}

\section{Appendix B}\label{B}
In this part we review the first result of the Section 2 in the case of random Hamiltonians. More precisely, we consider a 
$\z^d$-ergodic operator $H_\omega$. We adapt the notations and the quantities used previously. We consider the random 
$(\sigma,\zeta)$-subexponential moment

\begin{equation}\label{M_omega}
M_{u,\omega}(\sigma,\zeta,\f,t):= \normsch{\e^{{\frac{\sigma}{2}}|X-u|^{\zeta}} \e^{-itH_\omega} \f(H_\omega) \chi_u}^2 .
\end{equation}

We establish a similar version in expectation of Theorem~\ref{equiv1} that we formulate as

\begin{theorem}
 Let $I\subset\sigma(H)$ be an interval and assume that $H$ has pure point spectrum in $I$. The following properties are equivalent.
\begin{enumerate}
  \item[(i)] There exist $\sigma>0, \zeta\in(0,1]$ so that for any $\f\in\mathcal{C}^{\infty} _{0,+}(I)$, 
\begin{equation}\label{Eexp_DL}
\sup_T \M_u(\sigma,\zeta,\f,T)  :=\sup_T \frac{1}{T}\int_{0}^{T} \Es\{M_{u,\omega}(\sigma,\zeta,\f,t)\}  \dd t < \infty.
\end{equation}
 \item[(ii)] There exist $\sigma>0, \zeta\in(0,1]$ so that for any $\f\in\mathcal{C}^{\infty} _{0,+}(I)$,
\begin{equation}
\sup_T \frac1T \int_0^\infty \mathrm{e}^{-t/T}\Es\{M_{u,\omega}(\sigma,\zeta,\f,t)\}  \dd t < \infty.
\end{equation} 
\item[(iii)] There exist $\sigma>0, \zeta\in(0,1]$ any $\f\in\mathcal{C}^{\infty} _{0,+}(I)$, 
\begin{equation}\label{Eexp_DL0}
\Es( \sup_t M_{u,\omega}(\sigma,\zeta,\f,t)) < \infty.
 \end{equation}
  
\item[(iv)] There exist $\zeta\in (0,1], \sigma>0$ such that for 
for any $\f\in\mathcal{C}^{\infty} _{0,+}(I)$, there is a constant $C_{\zeta,\sigma,\f}<\infty$, so that
\begin{equation}\label{ESUDL} 
\Es(\sup_t\normsch{\chi_x \ \e^{-itH}\f(H) \chi_u})\leq C_{\zeta,\sigma,\f} \ \e^{-\sigma\abs{x-u} ^{\zeta}} \ \ 
for \ all \ x,u \in\z^d. 
\end{equation}

\item[(v)] There exist $\zeta\in (0,1], \sigma>0$ such that for any $\f\in\mathcal{C}^{\infty} _{0,+}(I)$, 
there is a constant $C_{\zeta,\sigma,\eps,\f}<\infty$, so that 
\begin{equation}\label{ESULP}
\Es(\sup_k\f(E_{k,\omega})\normsch{\chi_x P_{k,\omega} \chi_u})\leq C_{\zeta,\sigma,\f} \
 \e^{-\sigma\abs{x-u} ^{\zeta}} \ \ for \ all \ x,u \in\z^d. 
 \end{equation} 
\end{enumerate}
If $H$ satisfies one of these properties, we say that $H$ exhibits strong dynamical localization in $I$.

 \end{theorem}
The proof is similar to that of Theorem~\ref{equiv1} and notice that the ergodicity allows us to study the dynamics just from the origin 
($u=0$) through the moments. Furthermore, we should take the randomness in account and add it in all other quantities that we have
introduced. 

\begin{proof} 
Once again, the points that we should prove are $(i)\Rightarrow(v)\Rightarrow (iv)$.
As in \eqref{P_u} and \eqref{L_u}, we introduce

\begin{equation}\label{P_omega}
\p_\omega(x,\f):= \sup_k \f(E_{k,\omega})\normsch{\chi_x P_{k,\omega} \chi_0},
\end{equation}

\begin{equation}\label{L_omega}
\lp_\omega(\sigma,\zeta,\f):=\sum_{x\in\z^d} \e^{\sigma\abs{x}^{\zeta}} \ \p_\omega ^2 (x,\f),
\end{equation}
and 
\begin{equation}\label{L_expect}
\lp(\sigma,\zeta,\f):=\sum_{x\in\z^d} \e^{\sigma\abs{x}^{\zeta}} \Es\left(\p_\omega ^2 (x,\f)\right).
\end{equation}
Using the same strategies, we have
\begin{align}\label{}
\liminf_{T\to\infty} \frac{1}{T}\int_0^T M_{0,\omega}(\sigma,\zeta,\f,T) 
& \geq C_{\sigma,\zeta} \sum_k \sum_{x\in\z^d}  \f^2(E_{\omega,k}) 
\e^{\sigma\abs{x}^{\zeta}} \normsch{\chi_x P_{\omega,k} \chi_0}^2 \notag\\
& \geq C_{\sigma,\zeta} \sum_{x\in\z^d}  \e^{\sigma\abs{x}^{\zeta}}\left(\sup_k \f^2(E_{\omega,k})\normsch{\chi_x P_{\omega,k} \chi_0}^2
\right)\notag.
\end{align}
Taking the expectation, we obtain
\begin{equation}\label{expectation}
\Es\left(\liminf_{T\to\infty} \frac{1}{T}\int_0^T M_{0,\omega}(\sigma,\zeta,\f,T)\right) 
\geq C_{\sigma,\zeta} \sum_{x\in\z^d}  \e^{\sigma|x|^{\zeta}} \Es\left(\sup_k \f(E_{\omega,k})\normsch{\chi_x P_{\omega,k} \chi_0}
\right)^2,\\
\end{equation}
and the Fatou lemma yields
\begin{equation}\label{sulp_exp}
 \liminf_{T\to\infty} \M_0(\sigma,\zeta,\f,T)\geq C_{\sigma,\zeta} \ \lp(\sigma,\zeta,\f).
\end{equation}
Consequently, we get a similar result to \eqref{getsulp}. For the last point, we go back to Theorem~\ref{kernel} and Lemma~\ref{order}
 that we restore for $\omega$ fixed. Then for any $\gamma\in(0,1)$, there exists a constant $C_{\sigma,\zeta,d,\gamma}$ which is 
uniform in $\omega$ such that 
 
\begin{equation}
 \sup_t \normsch{\chi_x \ \e^{-itH_\omega} \f(H_\omega) \chi_0}\leq C_{\sigma,\zeta,d,\gamma} \ 
\p_\omega^{1-\gamma}(x,\f) \lp_\omega^{\gamma/2}(\sigma,\zeta,\f),\notag
\end{equation}
and hence
\begin{equation}
 \Es\left(\sup_t \normsch{\chi_x \ \e^{-itH_\omega} \f(H_\omega) \chi_0}\right)\leq C_{\sigma,\zeta,d,\gamma} \ 
\Es\left(\p_\omega(x,\f)\right)^{1-\gamma} \Es(\lp(\sigma,\zeta,\f))^{\gamma/2},\notag
\end{equation}
thanks to the H\"{o}lder inequality that we applay with conjugate exponents $p=\frac{1}{1-\gamma}$ and $p'=1/\gamma$ 
and to Jensen's inequality.

\end{proof}

\end{document}